\documentclass[11pt, letterpaper]{article}

\usepackage{CJKutf8}

\usepackage[whole]{bxcjkjatype}

\usepackage[pdftex]{graphicx,xcolor}
\usepackage[fleqn]{amsmath}

\usepackage{algorithm}
\usepackage[noend]{algorithmic}

\usepackage{cite} 
\usepackage{threeparttable}
\usepackage{amsthm} 
\usepackage{latexsym}
\usepackage{comment}
\usepackage{bm}
\usepackage{enumerate}
\usepackage{bookmark} 
\usepackage{type1cm} 
\usepackage{color}
\usepackage{comment}
\usepackage{multirow}

\usepackage{newtxtext}
\usepackage[varg]{newtxmath}

\usepackage{thm-restate}
\usepackage {threeparttable}
\usepackage{algorithmic}
\usepackage{algorithm}
\theoremstyle{plain}
\newtheorem{theorem}{Theorem}
\newtheorem{lemma}{Lemma}

\newcounter{cntLemmaNumber}

\newcounter{cntTheoremNumber}

\graphicspath{{./figure/}}

\title{Improved Algorithms for Local Failover Routing on Directed Graphs
\author{ Yuki Kawashima
\and Naoki Kitamura
\and Taisuke Izumi
}
} 

\begin{document}
\maketitle

\begin{abstract}
    The local failover routing is a mechanism that routes a packet from a source to a destination only using pre-calculated routing tables, even when several edges fail. In this paper, we study local failover schemes that minimize the number of rewritable bits in the packet header on directed graphs with $k$-arc failures. 
    There are many studies of failover routing on undirected graphs, and it has been investigated whether routing is possible depending on the number of bits in the packet header, the type of failure, the graph properties, etc.
    In contrast, there is not much research on directed graphs.
    Van et al.~first showed the upper and lower bounds of rewritable bits in the packet header on directed graphs. However, their results showed a large gap between the upper and lower bounds.
    The main contribution of this paper is to close the gap between the upper and lower bounds. Specifically, we show that our scheme can route packets with $k$ faulty arcs if the packet header has $\min(k \log ( \frac{e(2n+k-3)}{k}, 2n \log ( \frac{e(2n+k-3)}{2n})))$ rewritable bits, where $n$ is the number of nodes. Moreover, any local failover routing scheme needs $\Omega(k\lceil\log\frac{n}{k}\rceil)$ rewritable bits when the number of faulty arcs is equal to or less than $\frac{3(n-1)}{8}$ and $\frac{n-1}{4}$ rewritable bits when the number of faulty arc is more than $\frac{3(n-1)}{8}$. This result means our scheme is nearly optimal when the number of faulty arcs is approximately less than the number of nodes. 
\end{abstract}

\section{Introduction}
    The local failover routing is a mechanism that delivers a message from $s$ to $t$ only using pre-calculated routing tables, even when several edges fail (i.e., deleted from the network).
    It is known that static local failover routing from $s$ to $t$ in a directed graph is impossible even in the case of a single arc failure unless the packet has rewritable bits in its header~\cite{van2024brief}. Hence, as prior results, we also focus on the local failover routing with packet headers. Several studies have investigated whether message transmission is possible with respect to some parameters such as the number of faulty edges, the properties of the graph structure, and the upper limit of the rewritable bits in the packet header~\cite{van2024brief, dai2024resilience}.
    For the upper bound $k$ on the number of failing arcs, Van et al.~showed that a rewritable header of $k \lceil \log |E| \rceil$ bits suffices to deliver messages between any nodes, here $|E|$ is the number of arcs in the graph.
    On the other hand, the state-of-the art lower bound for header size is $\lceil\log (k+1)\rceil$ bits~\cite{van2024brief}, which has a large gap with the upper bound of Van et al.

    In this paper, we further investigate the upper and lower bounds of the rewritable bits required for static local failover routing in a directed graph.
    Precisely, our main contributions present tighter upper and lower bounds on packet header sizes in several settings, which are stated as follows:
    \begin{restatable}{theorem}{upperBoundKone} \label{thm:upper_k1}
        Let $G = (V, E)$ be a directed graph with $n$ nodes and $D$ be its diameter of $G$. 
        Assuming at most one arc failure, there exists a local failover routing scheme $\mathcal{A}$ that uses $\lceil \log (D + 1) \rceil$-rewritable bits in the packet header. 
    \end{restatable}
    
    \begin{restatable}{theorem}{upperBoundK} \label{thm:upper_k}
        For any $k$, let $G = (V, E)$ be a directed graph with $n$ nodes. 
        Assuming at most $k$ arc failure, there exists a local failover routing scheme $\mathcal{B}$ that uses at most
        \[
        \min(k \lceil\log \frac{e(2n+k-3)}{k}\rceil,2n \lceil\log \frac{e(2n+k-3)}{2n}\rceil)
        \]
        rewritable bits in the packet header.
    \end{restatable}

    \begin{restatable}{theorem}{lowerBoundKthree} \label{thm:lower_k3}
        There exists a lower bound graph $G$ of $n$ nodes such that any local failover routing tolerates at least three arc failures requires at least $\lceil \log (n - 1)\rceil - 1$ and $\lceil \log (n - 2)\rceil - 1$ rewritable bits in the packet header when $n$ is odd and even, respectively.
    \end{restatable}

    \begin{restatable}{theorem}{lowerBoundK} \label{thm:lower_k}
        For any integer $3 \leq k \leq \frac{3(n-1)}{8}$, there exists a lower bound graph $G$ of $n$ nodes such that any local failover routing tolerates at least $k$ arc failures requires at least $\Omega(k\lceil\log\frac{n}{k}\rceil)$ bits in the packet header.
    \end{restatable}

    \begin{restatable}{theorem}{lowerBoundMax} \label{thm:lower_max}
        For any integer $\frac{3(n-1)}{8} < k$, there exists a lower bound graph $G$ on the $n$ nodes such that any local failover routing for it and at least $k$ arc failures require at least $\frac{n-1}{4}$ bits in the packet header.
    \end{restatable}
    Table~\ref{tab:results} summarizes the previous results and our new contributions.
    For one failure, our upper bound replaces the worst-case path length $n-1$ with the diameter $D$. 
    Since the diameter $D$ is at most $n - 1$, our result of $\lceil \log (D + 1) \rceil$ is better than the existing $\log n$.
    
    For multiple failures, the improvement over the previous upper bound $k\lceil \log |E| \rceil$ is seen by comparing it with the two terms in Theorem~\ref{thm:upper_k}. The first term is roughly
    \[
        k\log\frac{n+k}{k}=k\log\left(1+\frac{n}{k}\right),
    \]
    which replaces the factor $\log |E|$ by $\log(1+n/k)$. Thus, when $k$ is small or moderate compared with $n$, the cost is about $k\log(n/k)$ rather than $k\log |E|$; the saving becomes larger as $k$ increases. 
    When $k$ is large, the second term,
    \[
        n\log\frac{n+k}{n}=n\log\left(1+\frac{k}{n}\right),
    \]
    can be smaller, so the minimum prevents the bound from growing linearly with $k\log |E|$. On the lower bound, Theorems~\ref{thm:lower_k3}--\ref{thm:lower_max} improve the previous lower bound $\lceil\log(k+1)\rceil$ by showing that some graphs require a header size that grows with $n$. Thus, for a wide range of $k$, the known upper and lower bounds become much closer.
    
    The theorems above successfully narrow the gap between the upper and lower bounds on the number of rewritable bits required for packet headers.
    \begin{table*}[h]
        \centering
        \caption{Summary of related work and our new contribution.}
        \label{tab:results}
        \scalebox{0.625}{
        \begin{tabular}{c c c c c}
            & Number of failures & Upper bound of rewritable bits & Lower bound of rewritable bits & Failure type\\
            \hline \hline
            Van et al.~\cite{van2024brief} &  $k=1$ & $\log n$ & $1$ & static\\
            Van et al.~\cite{van2024brief} &  $k\geq2$ & $k \lceil \log |E|\rceil$ & $\lceil \log (k+1)\rceil$ & static\\ \hline
            Theorem~\ref{thm:upper_k1} &  $k=1$ & $\lceil \log (D + 1) \rceil$ & - & static\\
            Theorem~\ref{thm:upper_k} & $k \geq 2$ & $\min(k \lceil\log \frac{e(2n+k-3)}{k}\rceil,2n \lceil\log \frac{e(2n+k-3)}{2n}\rceil)$ & - & static \\
            Theorem~\ref{thm:lower_k3} & $k=3$ & - & $\lceil \log (n-1)\rceil - 1$ & static \\
            Theorem~\ref{thm:lower_k} & $3 \leq k \leq \frac{3(n-1)}{8}$ & - & $\Omega(k\lceil\log\frac{n}{k}\rceil)$ & static \\
            Theorem~\ref{thm:lower_max} & $\frac{3(n-1)}{8} < k$ & - & $\frac{n-1}{4}$ & static
        \end{tabular}
        }
    \end{table*}

    \subsection{Related Works}
    There has been a lot of research on failover routing, in both directed and undirected networks. First, we will introduce some results on undirected graphs.
    Feigenbaum et al.~first presented theoretical study on local failover routing~\cite{feigenbaum2012brief}. They showed that, without header information, single-edge failures are always tolerable, but tolerating multiple failures is generally impossible.
    Dai et al.~investigated the limitation by Feigenbaum et al.~in more details. Precisely, they showed that the routing without rewriting packet headers is possible for two-edge failures, but impossible for three or more failures~\cite{dai2023tight}.
    Chiesa et al.~proposed a routing method that can tolerate $k$-$1$ failures using $\log k$ bits in $k$-edge-connected graphs for any $k \leq 5$~\cite{chiesa2016resiliency}.
    Foerster et al.~studied fault tolerance in two variants of the models in which vertices can and cannot identify the source of a packet. They showed that perfect fault tolerance is impossible for nonplanar graphs. They also proposed algorithms achieving perfect fault tolerance for all outerplanar graphs and related settings, as well as for nonouterplanar graphs where the destination is within two hops of the source~\cite{foerster2021feasibility}.
    Dai et al.~classified link failures into three types: static (i.e., links which permanently and simultaneously fail), semi-dynamic (removing the assumption that links fail simultaneously), and dynamic (removing the assumption that links fail permanently), and examined the fault tolerance of failover routing for each type, clarifying its capabilities and limitations~\cite{dai2024resilience}. As a result, they showed a routing method which tolerates $k-1$ dynamic link failures in $k$-edge-connected graphs for $k \leq 5$. Furthermore, they showed that this result can be extended to any $k$ by providing $\log k$-rewritable bits in the packet header. They also showed that rewriting $3$ bits suffices to cope with $k$ semi-dynamic failures. However, on general graphs, tolerating $2$ dynamic failures becomes impossible without rewritable bits. Even by rewriting $\log k$ bits, fault tolerance is impossible for $k$ dynamic failures.

    In contrast to undirected networks, the local failover routing for directed networks is just at the beginning.
    Grobe et al.~carefully examined whether existing methods for undirected graphs can be applied to directed graphs, and compared the existing methods with their proposed method through simulations. As a result, they showed that their application is superior to the existing algorithms in most topologies~\cite{grobe2024local}.
    Van et al.~analyzed real-world topologies and showed that two-edge failure tolerance is possible in many real networks using two additional bits~\cite{van2024short}. However, these two earlier results do not provide any  theoretical analyses.
    Van et al.~investigated local failover routing in directed graphs and gave upper and lower bounds on the number of bits required in the packet header~\cite{van2024brief}. They showed that at least one bit is required for one failure in the packet header and that $\log n$ bits are sufficient to tolerate one failure. They also showed that at least $\log (k+1)$ bits are required for $k > 1$ failures, and $k\log|E|$ bits are sufficient to tolerate $k$ failures.

    Therefore, previous work on local failover routing in directed graphs leaves a large gap between the upper and lower bounds on the number of rewritable bits in the packet header. This paper focuses on this gap in the static failure model. Compared with the known upper bound of Van et al., our routing schemes reduce the required header size by encoding failure information more compactly. On the lower bound, our constructions show that substantially more than the known lower bound are sometimes necessary. These upper and lower bounds show how many rewritable header bits are sufficient, and how many can be necessary, for storing failure information in directed local failover routing.

    \subsection{Organization}
        In Section~\ref{sec:model}, we explain the routing model.
        In Section~\ref{sec:upper}, we prove the upper bounds on header size for local failover routing.
        In Section~\ref{sec:lower}, we prove the lower bounds on header size for local failover routing.
        Finally, in Section~\ref{sec:conclusion}, we explain the conclusion and open problems.
        
\section{Model} \label{sec:model}
    We consider a simple directed graph $G = (V, E), |V|=n$, where nodes represent routers and arcs represent communication links. Each node is assigned a unique identifier. Each arc is also assigned with a unique identifier represented by the pair of its endpoints' IDs. The diameter of the graph $G$ is denoted by $D$.  
    A local failover routing $R(G,k)$ for $G = (V, E)$ and $k$ is defined as a collection of $|V|$ local forwarding rules assigned to vertices in $V$. 
    A local forwarding rule $R_v$ at $v$ is defined as the following function:
    \begin{align*}
        & R_v:E \times 2^{|\delta(v)|} \times V \times V \times \{0,1\}^{\ast} \to E \times \{0,1\}^{\ast}. \\
    \end{align*}
    Here, $\delta(v)$ is the set of outgoing arcs of $v$ (i.e., $\delta(v) = \{(v, u) \in E:u \in V\}$).
    $R_v$ determines how to process a packet arriving at the node. For $R_v(\mathit{in}, F_v, s, t, b) = (\mathit{out}, b')$, the given arguments and returned values mean: 
    \begin{itemize}
        \item $in$ : The incoming arc from which the packet arrived (or $\varepsilon$ if the packet is originated at the node).
        \item $F_v$ : The set of locally faulty outgoing arcs. 
        Note that if the node $v$ sends a message through a faulty arc, it finds the arc failed. Then that arc is added to $F_v$.
        \item $s$ : The ID of the source node of the packet. This information is stored in the packet header.
        \item $t$ : The ID of the target node of the packet. This information is stored in the packet header.
        \item $b$ : A rewritable bit string carried in the packet header.
        \item $out$ : The outgoing arc to which the packet is sent.
        \item $b'$ : The updated bit string carried in the header of the packet sent out.
    \end{itemize}
    Note that the IDs of $s$ and $t$ in the packet header are never rewritten.
    When a packet arrives at a node other than its destination $t$, the node decides the neighboring node to which that packet is forwarded following the rule $R_v$. After that, it sends the packet to the next node. When a node fails to send a packet, it stores the faulty arc in $F_v$ and rewrites the packet header. Then it retries the packet forwarding according to the forwarding rules.
    We model the directed local failover routing as a two-player game between player 1 (the designer) and player 2 (the adversary):
    \begin{enumerate}
        \item The players are given a simple directed graph $G = (V, E)$ with $n = |V|$ nodes and a failure parameter $k$.
        \item The player 1 defines a set of local forwarding rules for each node, which must be able to route packets from a given source $s$ to a target $t$ as long as a $s$-$t$ path exists in the current graph.
        \item Let $F$ be a subset of $E$ that contains at most $k$ arcs. The player 2 removes $F$ from $G$.
        \item For any $s$ and $t$ such that  every node reachable from $s$ in $G - F$ is reachable to $t$ in $G - F$, verify if the predefined local forwarding rules can successfully guides the packet from $s$ to$ t$ or not.
        If it succeeds for all those pairs, then player 1 wins. Otherwise player 2 wins.
    \end{enumerate}
    A local failover routing scheme $\mathcal{A}$ is an algorithm of outputting a local failover routing for a given network $G$ and a threshold parameter $k$.
    The efficiency of a local forwarding rule is measured by the size of rewritable bits used in designed forwarding rules (referred to as header size).
    When the header size of any packet that can be sent in routing $R(G,k)$ is less than or equal to $B$, we say that the header size of $R(G, k)$ is at most $B$. For any $n$-node directed graph $G$ and the failure parameter $k$, if the header size of $R(G,k)$ outputted by the scheme $\mathcal{A}$ is bounded by a function $f(n,k)$, the header size of the $\mathcal{A}$ is said to be $f(n,k)$.

\section{Upper Bounds on Header Size for Local Failover Routing} \label{sec:upper}
    In this section,  we present local failover routing schemes $\mathcal{A}_1$ and $\mathcal{A}_2$ respectively providing the upper bounds of Theorem~\ref{thm:upper_k1} and~\ref{thm:upper_k}.

    Since any packet contains the information of its source and destination, one can apply different  forwarding rules for packets with different source-destination pairs. Hence in the following argument we focus on the construction of the forwarding rule for a fixed pair $(s, t)$. The combination of the constructed rules for all $(s, t)$ obviously deduces whole forwarding rule.
    The baseline strategy  common among $\mathcal{A}_1$ and $\mathcal{A}_2$ is stated as follows:
    The routing algorithm tries to send the packet via a path from $s$ to $t$.
    If the packet encounters a faulty arc in the transfer along the path, it records the arc ID in its header. After that, the packet is transferred along a predetermined route in the input graph, provided that the set of faulty arcs carried in the packet header are all removed. Since we assume that $v$ must be reachable to $t$ if the packet is reachable to a vertex $v$ from $s$, there  necessarily exists a path to $t$ avoiding all faulty arcs. Hence one can always design a correct failover routing scheme following this strategy. The design factor lying in this strategy is twofold:  First, how the information of faulty arcs written to the packet header is encoded. Second, which path from $v$ to $t$ is chosen when a set of (identified) faulty arcs is given. 

    \subsection{Upper Bound for a Single Failure} \label{sec:thm1}
        In this section, we show Theorem~\ref{thm:upper_k1}.
        Since explicitly storing faulty arc IDs requires $k \log n$  bits in the worst case, we need to encode the information of the set of faulty arcs more compactly. Our key technical idea is to store the ``time'' when the node fails to send the packet\footnote{While the technique of recording the number of transfers has been utilized in several known literatures particularly in undirected settings~\cite{chiesa2016resiliency, bankhamer2022local}, their purpose appears to be loop avoidance. To the best of our knowledge, our result is the first scheme using such a recorded count for data recovery after failures.}.
        That is, instead of remembering the arc ID, the packet header remembers when the packet forwarding failed, counting from the origination of the packet. We refer to the time when packet transfer fails as a \emph{failing time}.
        Because packet transmission is deterministic, it is possible to determine the route where a packet was transferred by using the header information. In addition, it is possible to restore the set of the faulty arcs the packet encounters from its header.
        \upperBoundKone*
        \begin{proof}
            As we mentioned, it suffices to consider the design of the forwarding rule for a fixed source-destination pair $(s, t)$.
            Let $P = v_0, e_0, v_1, e_1, \dots v_d$ ($v_0 = s$ and $v_d = t$) be any shortest $s$-$t$ path in $G$, and $P_i$ be the shortest $v_i$-$t$ path in $G - e_i$. Note that such a path necessarily exists by the assumption of the problem definition. The header of each packet initially stores the value $d$, which means that the packet does not yet encounter faulty edges. We set up the forwarding rule of $v_0, v_1, \dots,  v_d$ so that the packet with header information $d$ is routed along $P$. If the packet forwarding through arc $e_i$ failed, the header information is updated with $i$. For each $0 \leq i \leq d-1$, we set up the forwarding rules of vertices in $P_i$ so that the packet with value $i$ is routed along it.

            Since we assume $k = 1$, this forwarding rule obviously delivers the packet from $s$ to $t$. The header size is bounded by $\lceil \log (d + 1) \rceil \leq \lceil \log (D + 1) \rceil$.
        \end{proof}
        
    \subsection{Upper Bound for Multiple Failures} \label{sec:thm2}
        In this section, we show Theorem~\ref{thm:upper_k}.
        The strategy for $\mathcal{A}_2$ is close to $\mathcal{A}_1$. The key idea behind $\mathcal{A}_1$ is the packet passes through at most $D$ arcs, and thus it can encounter at most $D$ different positions in its transfer. Hence $\log D +1$ bits suffices to store the information of the faulty arc (note that each forwarding rule is designed in advance with full information of $G$, and thus the information of $s$-$t$ shortest path (i.e., $P$ in the proof of Theorem~\ref{thm:upper_k1}) and other failure-avoiding paths $P_i$ are available).  The scheme $\mathcal{A}_2$ generalizes this approach to the multiple failure case.
        
        If the length of the path delivering packets from $s$ to $t$ is bounded by $g(k)$, then it can store the information of all failing times by using  $(\log \binom{g(k)}{k})$-rewritable bits.
        Unfortunately, a routing strategy that always tries to send the packet along a shortest path to $t$ in the remaining graph (i.e., the graph after removing all encountered faulty arcs) may traverse the $\Omega(kn)$ arcs in the worst case.
        To reduce the number of required bits, we consider the following strategy: An arc that succeeded to send the packet can be safely reused in the subsequent packet routing. Hence we do not have to count up the time when reusing such an arc into account, to identify encountered faulty arcs via failing times. This observation allows us to shrink the range of failing times into a value smaller than the transferred path length. To maximize this benefit, when encountering a faulty arc, we select the path from the current node to $t$ that uses the fewest untraversed arcs, in the graph with removing all encountered faulty arcs. 
        
        We explain our scheme in more details. Our scheme (for each $s$-$t$ pair) consists of the two components $\mathrm{dec}(b, v)$ and the set of paths $P_{F, v}$ for any $F \subseteq E$ of $|F| \leq k$ and $v \in V$. The routing of the packet with information $b$ at $v$ first decodes the rewritable bits $b$ using $\mathrm{dec}(\cdot)$, which returns the set of faulty arcs $F'$ the packet encountered up to $v$. Then the algorithm forwards its packet along the specified path $P_{F', v}$.
        In our scheme, the packet header stores the optimally compressed form of the \emph{failure profile} $c$, which is defined as follows: Let $P = e_1, e_2, \dots, e_h$ be the forwarding edge sequence such that the packet is sent through $e_i$ at its $i$-th forwarding (note that $e_i$ might be faulty and thus it is not guaranteed that $P$ forms a path from $s$ to $v$). The \emph{first edge sequence} $\tau(P)$ of $P$, which is the one obtained from $P$ by extracting the first appearance of each edge. For example, given $P = e_1, e_3, e_2, e_1, e_3, e_4$, $\tau(P)$ becomes $e_1, e_3, e_2, e_4$. The failure profile $c$ is defined as the bit string of length $|\tau(P)|$ such that the $i$-th bit of $c$ is one if and only if the $i$-the edge of $\tau(P)$ is faulty.
        
        The function $\mathrm{dec}(b, v)$ first uncompresses $b$ into $c$, and
        simulates the packet forwarding from $s$ to $v$ under the condition of the failure events represented by $c$. This simulation is possible because $v$ knows whole information on $G$. Obviously, it also reconstructs $P$, and thus provides the information on $F'$. To minimize the size of header information, our scheme chooses $P_{F',v}$ such that the length of $\tau(P \circ P_{F',v})$ is minimized (where $\circ$ means the concatenation of two sequences). More precisely,
        we choose $P_{F',v}$ as the shortest $v$-$t$ path in the graph 
        $G - F'$ where weight zero is given to edge $e$ if it is contained in $P$, or one otherwise. 
        When the packet is sent through edge $e''$, if $e''$ is not contained in $P$, an additional bit 1 is appended to $c$, and its header is rewritten by the compressed form of $c$ after appending. If the packet transfer fails, $e''$ is added to the reconstructed $F'$, and the appended bit is changed to zero. Then the new route $P_{F', e''}$ is chosen. If $e''$ is
        the edge in $P$ (and not in $F$), the packet is transferred without any modification. 

        Now, we show the number of required bits for the rewritable bits under our strategy. Here, we show the following lemma.

        \begin{lemma} \label{lma:tryEdgeNum}
            In local $s$-$t$ routing with $k$ faulty arcs, the number of arcs that must be tried to send the packet to $t$ is at most $2n+k-3$ in our routing strategy. 
        \end{lemma}
        \begin{proof}
            The arcs which tries to transfer the packet from $s$ to $t$ can be classified into the following three types:
            \begin{enumerate}
                \item Arcs used for sending the packet to an unvisited node,
                \item Arcs used for sending the packet to an already visited node,
                \item Arcs that could not send the packet due to failure.
            \end{enumerate}
            The example of these three arc types is Figure~\ref{fig:arcTypes}.
            
            The number of type~(1) arcs is at most $n - 1$. The number of type~(3) arcs is at most $k$. 
            Now we consider the number of type~(2) arcs. 
            When a packet can be sent to an already visited node using a type~(2) arc, a directed cycle is formed by type~(1) arc(s) and a type~(2) arc.
            Since the directed cycle consists only of arcs that have been used in the previous routing, every node in the cycle is reachable by using only those previously used arcs.
            Moreover, no arc with both endpoints in the cycle is used in subsequent routing unless it belongs to the directed cycle.
            Therefore, the vertices included in the cycle can be considered as a single node in subsequent routing. 
            When the cycle is contracted into a single node, the number of nodes in the graph decreases by at least one, so the number of type~(2) arcs is at most $n-2$.

            Therefore, the number of arcs that must be tried to send the packet to $t$ is at most $2n+k-3$.
        \end{proof}
        
        \begin{figure}[h]
        \begin{center}
            \includegraphics[keepaspectratio, scale=0.7]{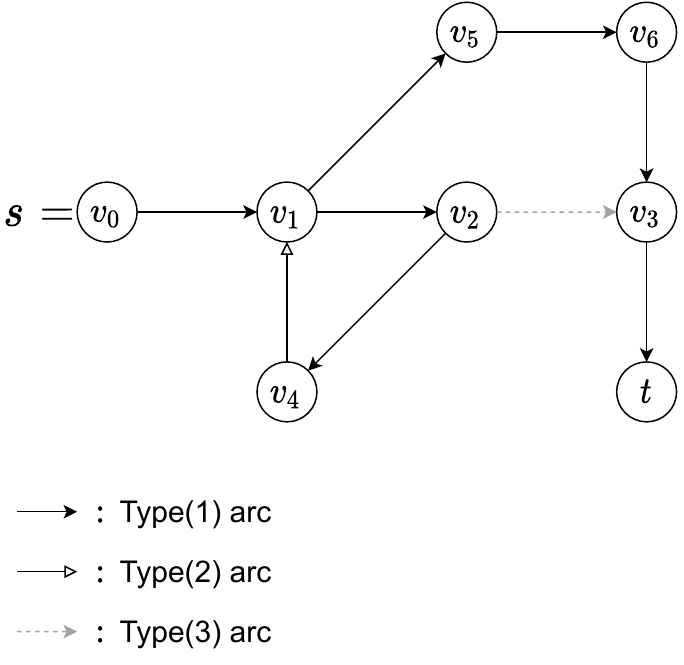} \\
            \caption{An example of arc types. In this case, the forwarding rule begins sending the packet from $s=v_0$. It sends the packet along the path $(v_0, v_1, v_2, v_3, t)$. When it tries to send to $v_3$, an arc $(v_2,v_3)$ is failed. So, it sends the packet along the path $(v_2, v_4, v_1, v_5, v_6, v_3, t)$.}
            \label{fig:arcTypes}
        \end{center}
        \end{figure}
        
        Now we show Theorem~\ref{thm:upper_k}.
        \upperBoundK*        
        \begin{proof}
            From Lemma~\ref{lma:tryEdgeNum}, it is sufficient to try sending the packet at most $2n+k-3$ times.
            In our strategy, the number of bits required for recording is $2n+k-3$ because the packet selects the path that is minimum number of unsent arcs before when the node selects the path. The number of faulty arcs is at most $k$, and when the number of faulty arcs detected is $l(<k)$, the number of $1$s in the recorded bit string can always be $k$ by setting the last $k-l$ digits of the recorded bits to $1$. There are a total of $\binom{2n+k-3}{k}$ ways for a bit string with $k$ $1$s out of $2n+k-3$ digits, so when writing this into the packet header, it can be encoded as a value of $(\log \binom{2n+k-3}{k})$ bits.
            Recall that if the information about the removed arcs in previous routing is stored in the packet header, it is possible to determine the route where a packet was sent. Furthermore, it can construct the forwarding route which does not contain a directed cycle on a graph with $k'$ ($k'<k$) arcs removed.
            Therefore, the number of times the packet transmission has been attempted can be calculated from the information in the packet header and on the current node. Therefore, $(\log \binom{2n+k-3}{k})$-rewritable bits are sufficient.
            Using Stirling's formula, we obtain the following inequality:
            \begin{align*}
                \log \binom{2n+k-3}{k}& \leq \log \left( \frac{(2n+k-3)^k}{(\frac{k}{e})^k}\right)\\
                &= k \log \left( \frac{e(2n+k-3)}{k}\right),
            \end{align*}
            where $e$ is Napier's constant.
            Similarly, the following inequality also holds:
            \begin{align*}
                \log \binom{2n+k-3}{k}&= \log \binom{2n+k-3}{2n-3}\\
                &\leq \log \left( \frac{(2n+k-3)^{2n}}{(\frac{2n}{e})^{2n}}\right)\\
                &= 2n \log \left( \frac{e(2n+k-3)}{2n}\right).
            \end{align*}
            Therefore, the number of bits required for a packet header is $\min(k \lceil\log \frac{e(2n+k-3)}{k}\rceil$, $2n \lceil\log \frac{e(2n+k-3)}{2n}\rceil)$.
        \end{proof}

\section{Lower Bound for bit-string length} \label{sec:lower}
    In this section, we show the lower bounds on the number of rewritable bits in the packet header for $k\geq 3$.
    To prove this, we construct a graph in which packet headers must encode a large number of failure patterns which must be distinguished for correct delivery of packets.
    First, we show the construction of a gadget $A_{\ell}$ as follows:
    \begin{itemize}
        \item Create two paths $(u_0,u_1,\dots,u_{\ell-1})$ and $(v_{0},v_1,\dots,v_{\ell-1})$.
        \item  Add an arc $(u_j,v_j)$ for each $0\leq j \leq \ell -1$.
    \end{itemize}
    An example of $A_\ell$ is illustrated in Figure~\ref{fig:uv}. 
    
    \begin{figure}[h]
    \begin{center}
        \includegraphics[keepaspectratio, scale=0.9]{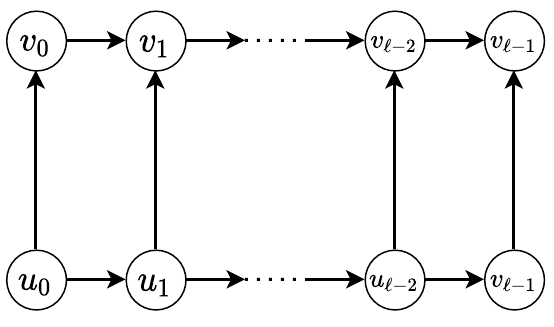} \\
        \caption{The example of the gadget $A_\ell$.}
        \label{fig:uv}
    \end{center}
    \end{figure}
    
    Here, we show the following lemma.
    \begin{lemma} \label{lma:remove3arcs}
        Prepare a vertex $r$ and a gadget $A_{\ell}$. Let $s = r$ and $t = v_{\ell-1}$, and add arcs $(u_i,r)$, $(v_i,r)$ for $0\leq i \leq \ell-1$ and $(r,u_0)$. 
        An example of the player 2 strategy is shown in Figure~\ref{fig:lowerbound-k3}.
        Let $k = 3$. Now, the player 2 can decide a set of faulty arcs so that any surviving $s$-$t$ path necessarily contains $(u_j, v_j)$ for any $0 \leq j \leq \ell-1$. Precisely, if the player 2 wants to enforce $s$-$t$ paths containing $(u_j, v_j)$, it breaks the following arcs: 
        \begin{itemize}
            \item $(u_{j+1}, u_{j+2})$ if $0 \leq j\leq  \ell-3$.
            \item $(u_{j+1}, v_{j+1})$ if $0\leq j\leq  \ell-2$. 
            \item $(v_{j-1},v_{j})$ if $1 \leq j \leq \ell-1$.
        \end{itemize}
        In this case, at least $\lceil \log \ell\rceil$ bits are required as rewritable bits in the packet header.
    \end{lemma}
    \begin{proof}
        As long as the player 2 follows the strategy described above, since every node has an arc to $s$ and there is a path from $s$ to $t$, any node reachable from $s$ can still reach $t$ even after removing the arcs.
        Removing $(u_{j+1}, u_{j+2})$ and $(u_{j+1}, v_{j+1})$ disconnects the paths from $s$ to $t$ through $(u_{j'},v_{j'})$ for any $j<j'$.
        Removing the third arc prevents the path from $s$ to $t$ through $(u_{j'},v_{j'})$ for any $j'<j$.
        As a result, the only valid way to reach $t$ is through the arc $(u_j,v_j)$\footnote{In local failover routing, the next hop is determined using the set of locally faulty outgoing arcs and the packet header. Therefore, if the packet at $u_j$ attempts to use the arc $(u_j,u_{j+1})$ and this arc is faulty, node $u_j$ can detect this locally without using the packet header. Thus, routing can be performed simply by attempting to move along the arc $(u_j,v_j)$, even without the information in the packet header. To prevent this header-independent forcing, both $(u_{j+1},u_{j+2})$ and $(u_{j+1},v_{j+1})$ must be failed to return the packet to $s$. If both arcs are failed, then when the packet arrives at $u_j$, the forwarding rule must decide whether to use the arc $(u_j,v_j)$. Since the packet at $u_j$ cannot observe the set of locally faulty outgoing arcs of $u_{j+1}$, this decision must be made using the packet header.}.

        Suppose for contradiction that there exists an algorithm that can route the packet using at most $\lceil \log \ell \rceil - 1$ rewritable bits.
        From the bound of the rewritable bit length, in this algorithm, the packets can traverse at most $\ell-1$ distinct paths between $s$ and $t$. Therefore, there exists an arc $(u_j, v_j)$ that cannot be traversed regardless of the header contents. If the player 2 adopts a strategy such that only the path that passes through this arc can reach $t$, then the packet routing between $s$ and $t$ will always fail, which is a contradiction.
    \end{proof}
    
    \begin{figure}[h]
    \begin{center}
        \includegraphics[keepaspectratio, scale=0.65]{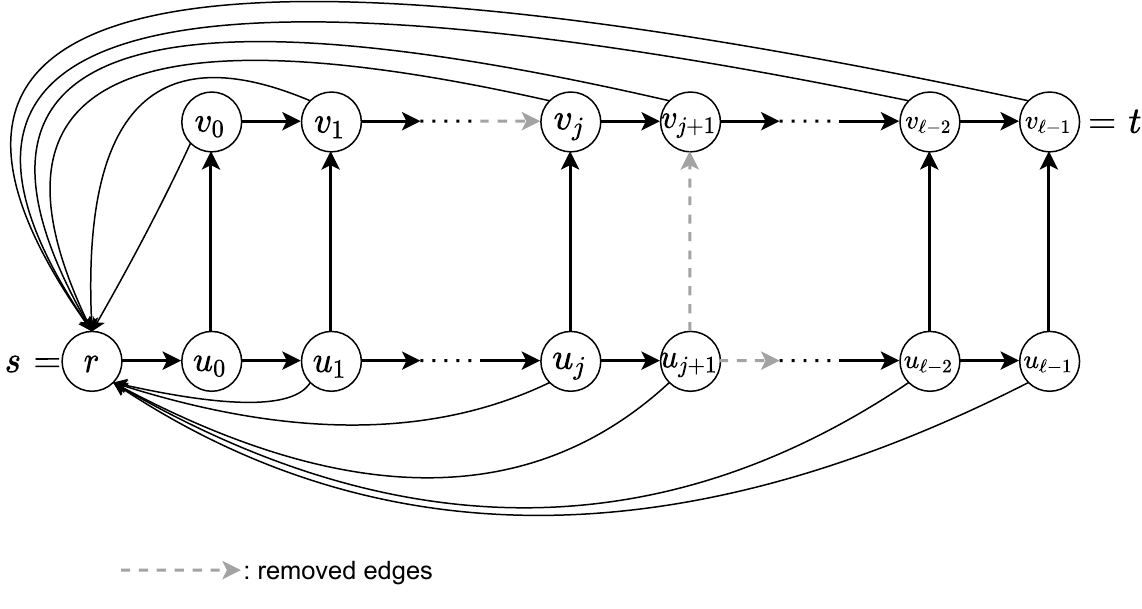} \\
        \caption{The example of removing arcs in Lemma~\ref{lma:remove3arcs}. Only the path through $(u_j,v_j)$ is a valid route to $t$.}
        \label{fig:lowerbound-k3}
    \end{center}
    \end{figure}
    
    \lowerBoundKthree*    
    \begin{proof}
        First, we consider when $n$ is odd.
        Let $N = ( n - 1 ) / 2$.
        The construction of a graph $G$ consists of a node $r$ and a gadget $A_{N}$.
        We add arcs $(u_i,r)$, $(v_i,r)$ for $0\leq i \leq N-1$ and $(r,u_0)$.
        Let $s = r$ and $t=v_{N-1}$.
        If the player 2 uses the same strategy as in Lemma~\ref{lma:remove3arcs} to remove arcs, then at least $\log N = \log\frac{n - 1}{2} = \lceil \log (n - 1)\rceil - 1$ rewritable bits are needed.

        Next, we consider when $n$ is even.
        Let $N = ( n - 2 ) / 2$.
        The construction of a graph $G$ consists of a node $q$, $r$ and a gadget $A_{N}$.
        We add arcs $(u_i,q)$, $(v_i,q)$ for $0\leq i \leq N-1$, $(q, r)$ and $(r,u_0)$.
        Let $s = r$ and $t=v_{N - 1}$.
        If the player 2 uses the same strategy as in Lemma~\ref{lma:remove3arcs} to remove arcs, then at least $\log N = \log\frac{n - 2}{2} = \lceil \log (n - 2)\rceil - 1$ rewritable bits are needed.
    \end{proof}
    
    Next, we construct a lower bound graph for general $k$ by cascading multiple copies of the gadget
 $A_{\ell}$.
    By Lemma~\ref{lma:remove3arcs}, we can limit the number of paths from $u_0$ to $v_{\ell -1}$ in the gadget $A_{\ell}$ to only one by removing three arcs. In this case, $\log \ell$ bits are required to remember which $(u_j,v_j)$ in the gadget is the path from $u_0$ to $v_{\ell-1}$. Thus, the total amount of information that must be encoded in the packet header grows in proportion to the number of connected gadgets.
    Now, we show Theorem~\ref{thm:lower_k}.
    
    \lowerBoundK*
    \begin{proof}
        Let $N=3(n-1)/2k$. 
        The lower bound graph $G$ consists of a node $r$ and $k / 3$ gadgets $A_{N}$ as $A_N^0, A_N^1, \dots, A_N^{\frac{k}{3}-1}$. Here, we denote any nodes $u_i$ and $v_i$ in $A_{N}^j$ by $u_i^j$ and $v_i^j$, respectively.
        Note that in our construction, we requires that the number of nodes in each gadget is greater than or equal to $8$. Hence, our construction only holds for the case of $k \leq \frac{3(n-1)}{8}$.
        We add arcs as follows:
        \begin{itemize}
            \item Arcs $(u_i^j,r)$, $(v_i^j,r)$ for $0 \leq i \leq N-1$ and $0\leq j \leq \frac{k}{3}-1$.
            \item Arcs $(v_{N-1}^{i},u_{0}^{i+1})$ for $0 \leq i \leq \frac{k}{3}-2$.
            \item An arc $(r, u_0^0)$.
        \end{itemize}
        Let $s=r$ and $t=v_{N-1}^{\frac{k}{3}-1}$.
        For each gadget $A^{j}$, the player 2 determines the index $i_j$ ($0\leq i_j\leq N-1$).
        Using a technique similar to Lemma~\ref{lma:remove3arcs}, for each gadget $A^{j}$, it removes at most three arcs such that only paths through the arc $(u^{j}_{i_j},v^{j}_{i_j})$ can reach $t$ from $s$.
        Note that player 2 can determine the index $i_j$ of each gadget after looking at the routing strategy of player 1. In addition, the indices of each gadget are selected independently.
        Therefore, player 1 must design a routing strategy that tries at least $N^{k/3}$ routes dependent on the packet header information. If the correct route is not selected, the packet is necessarily returned to $s$.
        Therefore, at the time of sending a packet from $s$ to $u_{0}^0$, the route of the packet determined by the header information must be correct eventually (recall that the packet header information is not updated unless it encounters a faulty arc, and thus the packet is not adaptively routed as long as it is correctly transferred). To decide which route to use, at least $\frac{k}{3}\log N = \frac{k}{3}\log \frac{3(n-1)}{2k} \in \Omega(k\log\frac{n}{k})$ rewritable bits are needed.
    \end{proof}
    \begin{figure}[t]
    \begin{center}
        \includegraphics[keepaspectratio, scale=0.65]{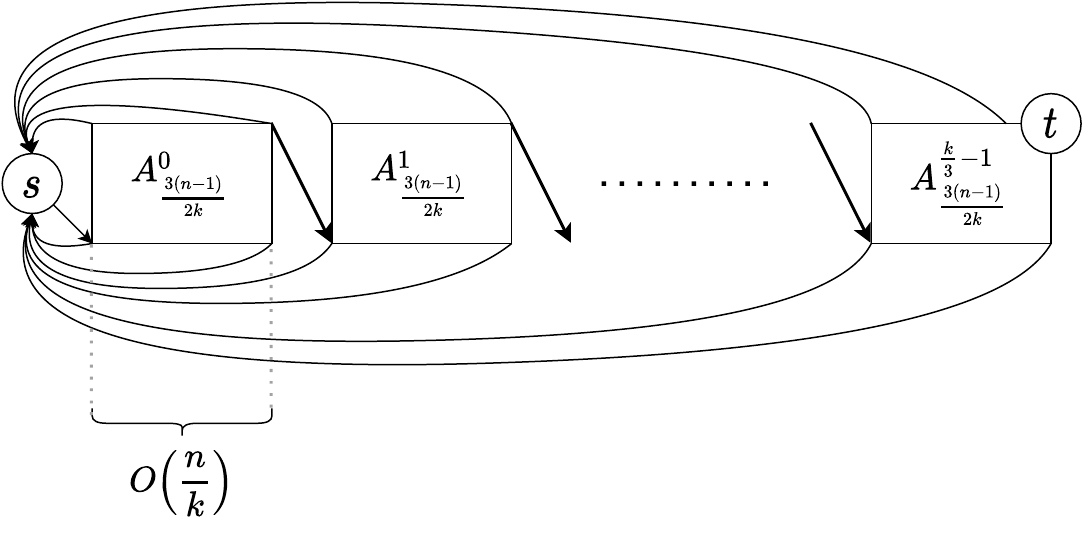} \\
        \caption{The example of the lower bound construction for general $k$. The graph consists of $k/3$ disjoint gadgets. Player 2 removes at most three arcs from each gadget.}
        \label{fig:lowerbound-k}
    \end{center}
    \end{figure}

    \lowerBoundMax*
    \begin{proof}
        Here, when $k = \frac{3(n-1)}{8}$, the lower bound graph in Theorem~\ref{thm:lower_k} requires $\frac{n-1}{4}$ bits in the packet header. 
        The player 2 can also remove only $\frac{3(n-1)}{8}$ arcs.
        Therefore, we can prove that the lower bound of the packet header size is $\frac{n-1}{4}$ bits if $k > \frac{3(n-1)}{8}$ by using the same lower bound graph of Theorem~\ref{thm:lower_k}. Therefore, Theorem~\ref{thm:lower_max} holds.
    \end{proof}

\section{Conclusion and Open Problems} \label{sec:conclusion}
    This study shows the upper and lower bounds of the number of rewritable bits required in the packet header for the local failover routing on directed graphs. As a result, we presented the local failover routing scheme that can route packets using only $\min(k \log ( \frac{e(2n+k-3)}{k},2n \log ( \frac{e(2n+k-3)}{2n})))$ rewritable bits. Moreover, any local failover routing scheme needs $\Omega(k\lceil\log\frac{n}{k}\rceil)$ rewritable bits when the number of faulty arcs is equal to or less than $\frac{3(n-1)}{8}$ and $\frac{n-1}{4}$ rewritable bits when the number of faulty arc is more than $\frac{3(n-1)}{8}$.
    Our results proposed a near optimal routing scheme when the number of faulty arcs is approximately less than the number of nodes. However, we have not yet shown the optimal routing scheme in which most of arcs are faulty.
    It is also interested in considering whether routing is possible for each graph property and type of failures.


\bibliography{reference}


\end{document}